\documentclass[]{interact}

\usepackage{mybinomcommon}

\begin{document}


\title{Parallelizing Computation of Expected Values in Recombinant Binomial Trees}

\author{
\name{Sai~K. Popuri\textsuperscript{a}\thanks{CONTACT Sai~K. Popuri. Email: saiku1@umbc.edu}, Andrew~M. Raim\thanks{Disclaimer: This article is released to inform interested parties of ongoing research and to encourage discussion of work in progress. Any views expressed are those of the authors and not necessarily those of the U.S. Census Bureau.}\textsuperscript{b}, Nagaraj~K. Neerchal\textsuperscript{a}, and Matthias~K. Gobbert\textsuperscript{a}}
\affil{\textsuperscript{a}Department of Mathematics and Statistics, University of Maryland, Baltimore County, 1000 Hilltop Circle, Baltimore, MD 21250, USA; \textsuperscript{b}Center for Statistical Research \& Methodology, U.S. Census Bureau, 4600 Silver Hill Road, Washington, DC 20233, USA}
}

\maketitle

\begin{abstract}
Recombinant binomial trees are binary trees where each non-leaf node has two child nodes, but adjacent parents share a common child node. Such trees arise in finance when pricing an option. For example, valuation of a financial option can be carried out by evaluating the expected value of asset payoffs with respect to random paths in the tree. In many variants of the option valuation problem, a closed form solution cannot be obtained and computational methods are needed. The cost to exactly compute expected values over random paths grows exponentially in the depth of the tree, rendering a serial computation of one branch at a time impractical. We propose a parallelization method that transforms the calculation of the expected value into an embarrassingly parallel problem by mapping the branches of the binomial tree to the processes in a multiprocessor computing environment. We also discuss a parallel Monte Carlo method which takes advantage of the mapping to achieve a reduced variance over the basic Monte Carlo estimator. Performance results from \proglang{R} and \proglang{Julia} implementations of the parallelization method on a distributed computing cluster indicate that both the implementations are scalable, but \proglang{Julia} is significantly faster than a similarly written \proglang{R} code. A simulation study is carried out to verify the convergence and the variance reduction behavior in the parallel Monte Carlo method.
\end{abstract}


\section{Introduction}
\label{sec:intro}

An $N$-step recombinant binomial tree is a binary tree where each non-leaf node has two children, which we will label ``up'' and ``down''. The tree has depth $N$, so that any path from the root node to a leaf node consists of $N$ up or down steps. The tree is called recombinant because the sequence of moves (up, down) is assumed to be equivalent to the sequence (down, up). In such a tree, there are $N+1$ distinct leaf nodes and $1 + 2 + \cdots + (N+1) = (N+1)(N+2)/2$ nodes overall. Any particular path from the root to a leaf can be written as a binary sequence $\vec{x} = (x_1, \ldots, x_N)$ where $x_j \in \mathbb{B}$, $\mathbb{B} = \{0, 1\}$, and 1 corresponds to an up movement while 0 corresponds to down. Given a density $p(\vec{x}) = \Prob(\vec{X} = \vec{x})$, we may consider $\vec{X}$ as a random path from the root to a leaf. We will refer to random variables $\vec{X} \in \mathbb{B}^N$ as Bernoulli paths. An $N$-step binomial tree has $2^N$ Bernoulli paths. 

A primary example of recombinant binomial trees is the binomial options pricing model proposed by Cox et al. \cite{CoxRossRubinstein1979}. This model accounts for uncertainty of a future stock price based on its current market price at $S$. Figure \ref{fig:2step} illustrates a binomial options model for the evolution of the stock in $N = 2$ time periods. Starting from the root node, the stock price moves up by an amount $u$ to $Su$ with probability $p$ or moves down to $S/u$ with probability $1-p$. After one step, each of the two child nodes further branch to two leaf nodes where a factor of $u$ is applied with probability $p$ or $d$ is applied with probability $1-p$. Here, the sequences (up, down) and (down, up) both take the stock price back to its starting price.

\begin{figure} \centering
	\includegraphics[width=0.65\textwidth]{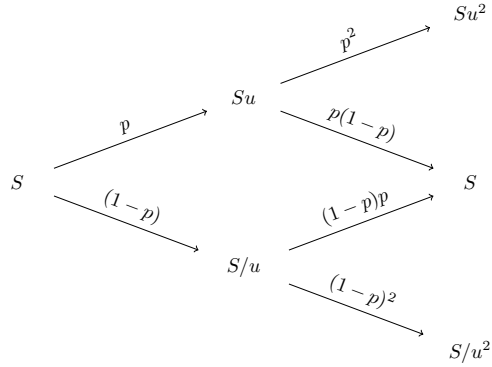}
	\caption{A two-step recombinant binomial tree.}
	\label{fig:2step}
\end{figure}

The binomial options pricing model is used in the valuation of financial contracts like options, which derive their value from a less complicated, underlying asset such as a stock price. In order to calculate the value of an option, one builds a recombinant binomial tree to a future time point from the current market price of the stock $S$ using a Bernoulli probability model at each time step. Depending on the type of the option, the option value is either the present value of the expected option payoff or is calculated by traversing the tree backwards and revising the option value at each step. See Hull \cite{Hull} and Seydel \cite{Seydel} for more details on options and their valuation. When the option payoff at a leaf node depends on the path, one must consider all $2^N$ possible paths to calculate the expected value of the option payoff. 

Pattern-mixture models for missing longitudinal data provide a second example involving recombinant binomial trees. A brief overview is given here, while the remainder of the paper focuses on the options pricing application. In a pattern-mixture model \cite{Little1993}, longitudinal data with missing values is available for each subject and the conditional distribution of the data given the pattern of missingness is considered. Let $Y_{it}$ be the response from subject $i$ at time $t$, where $i = 1, \ldots, n$ and $t = 1, \ldots, T$. The multivariate response $\vec{Y}_i = (Y_{i1}, \ldots, Y_{iT})$ may contain missing data whose  pattern is denoted by $\vec{Z}_i = (Z_{i1}, \ldots, Z_{iT})$; $Z_{it}$ is $0$ if $Y_{it}$ is observed and $1$ if missing. Hosseini and Neerchal \cite{HosseiniJSM2016} have adapted this framework to gerontological studies where caregivers provide responses on behalf of patients on some occasions, and patients themselves respond at other times. 
The joint distribution of the observed $\{ (\vec{Y}_i, \vec{Z}_i) : i = 1, \ldots, n \}$ for such a model is given by
\begin{equation}
\prod_{i=1}^n f(\vec{y}_i\mid \vec{z}_i,\vec{\theta})g(\vec{z}_i\mid \vec{\theta}),
\label{eq:pattern-mixture-likeli}
\end{equation}
where $f$ and $g$ are the probability functions of $\vec{y}_i\mid \vec{z}_i$ and $\vec{z}_i$, respectively. Note that the expected 
value calculations with respect to $\vec{z}$ will involve summing over all Bernoulli paths $\vec{z}$.

In applications of recombinant binomial trees, such as the two previously mentioned, it is often required to compute the expected value of a function $V(\vec{X})$
\begin{equation}
\E[V(\vec{X})] =
\sum_{\vec{x}\in \mathbb{B}^N} V(\vec{x})p(\vec{x}).
\label{eq:expval}
\end{equation}
The option value calculation and the pattern-mixture likelihood \eqref{eq:pattern-mixture-likeli} both take this form. The function $V(\vec{x})$ may depend on the entire path $\vec{x}$, and not only on the leaf nodes. Notice that \eqref{eq:expval} is a summation over $2^N$ terms, so that computing by complete enumeration quickly becomes infeasible as $N$ increases. In this work, we present a method to parallelize the calculation in a multiprocessor computing environment. In the option valuation 
problem, the common method to value options is to use an efficient backward induction method without 
considering the $2^N$ terms in \eqref{eq:expval}. The proposed parallelization method is suitable for advanced class of 
path-dependent options that are valued by sampling paths off the recombinant binomial tree than through backward 
induction \cite[Chapter 4.]{Glasserman2003}. 
Our method uses a Single Program Multiple Data (SPMD) approach \cite{Pacheco}, where each of the $M$ processes determines its assigned subset of $\mathbb{B}^N$ without coordination from a central process. Hence, the calculation can be transformed into an embarrassingly parallel problem \cite{Foster1995}, in which processes need not communicate except at the end of the computation, thereby allowing efficient scaling to many processes.
Even with a large number of processes $M$, the number of paths $2^N$ quickly becomes exceedingly large as $N$ increases. Therefore, we consider a Partitioned Monte Carlo method which uses a similar parallelization to reduce approximation error relative to basic Monte Carlo. 

The rest of the paper is organized as follows. Section \ref{sec:binomial} introduces the binomial tree model to value an option using Bernoulli paths. Section \ref{sec:parallel} describes a parallel scheme to compute the expected value exactly. Section \ref{sec:montecarlo} presents the Partitioned Monte Carlo method to approximately compute the expected value. Section \ref{sec:results} presents results from the implementation of the methods for put options in \proglang{R} and \proglang{Julia}. Concluding remarks are given in Section \ref{sec:conclusion}.

\section{Valuation of a path-dependent option using the binomial tree model}
\label{sec:binomial}
An option is a financial contract that gives the owner the right, but not the obligation, to either buy or sell a certain number of shares at a prespecified fixed price on a prespecified future date. A call option gives the owner the right to buy shares, while a put option gives the owner the right to sell shares. Several factors are used to value an option. The strike price $K$ is a prespecified fixed price. The time $T$ is the future date of maturity; for European options which are considered in this paper, the option can only be exercised at time $T$ and subsequently becomes worthless. The value of an option is the amount a buyer is willing to pay when the option is bought. It depends on $K$, $T$, and the characteristics of the underlying stock. More formally, let $V(S_t)$ denote the value of the option at time $t$, at which time the price of the underlying stock is $S_{t}$. We assume that time starts at $t=0$ at which point the option is bought or sold. The objective is to calculate $V(S_0)$, the value of the option at time $t=0$. Although $V(S_t)$ for $t < T$ is not known, the value $V(S_T)$, called the payoff, is known with certainty. The value $V(S_T)$ of a call option at the time of maturity $T$ is given by
\begin{equation}
\label{eq:call}
V(S_T) = \max\{S_T-K,0\}.
\end{equation}
For a put option, the value at the time of maturity $T$ is given by
\begin{equation}
\label{eq:put}
V(S_T) = \max\{K-S_T,0\}.
\end{equation}
Note that in \eqref{eq:call} and \eqref{eq:put}, the payoffs $V(S_T)$ depend only on the price of stock at time $T$, $S_T$, and the 
strike price $K$. In more complicated options, the payoffs often depend on additional factors. For example, the payoffs in  
path-dependent options depend on the historical price of the stock in a certain time period. For now, we will restrict our attention 
to simple options with payoffs in \eqref{eq:call} and \eqref{eq:put}.

The binomial tree method of option valuation is based on simulating an evolution of the future price of the underlying stock between $t=0$ and $T$ using a recombinant binomial tree. We first discretize the interval $[0,T]$ into equidistant time steps. We select $N$ to be the number of time steps, which determines the size of the tree, and let $\delta t = T/N$ be the size of each time step. Denote $t_i = i \, \delta t$ for $i = 0, \ldots, N$ as the distinct time points. Imagine a two-dimensional grid with $t$ on the horizontal axis and stock price $S_t$ on the vertical axis; by discretizing time, we slice the horizontal axis into equidistant time steps. We next discretize $S_{t}$ at each $t=t_i$ resulting in values $S_{t_ij}$, where $j$ is the index on the vertical axis. For notational convenience, we will write
$S_{t_ij}$ as $S_{ij}$. The binomial tree method makes the following assumptions.
\begin{description}
	\descitem{A1} The stock price $S_{t_i}$ at $t_i$ can only take two possible values over time step $\delta t$: price goes up to $S_{t_i}u$ or goes down to $S_{t_i}d$ at $t_{i+1}$ with $0<d<u$ where $u$ is the factor of upward movement and $d$ is the factor of downward movement. To enforce symmetry in the simulated stock prices, we assume $ud=1$.
	
	\descitem{A2} The probability of moving up between time $t_i$ and $t_{i+1}$ is $p$ for $i=0,\ldots,N-1$. 
	
	\descitem{A3} $\E(S_{t_{i+1}} \mid S_{t_i}) = S_{t_i}e^{q\delta t}$, where $q$ is the annual risk-free interest rate. For example, $q$ may be the interest rate from a savings account at a high credit-worthy bank.
\end{description}
Under assumptions A\ref{A1}--A\ref{A3}, and if the stock price movements are assumed to be lognormally distributed with 
variance $\sigma^2$, it can be shown that
\begin{align*}
&u = \beta + \sqrt{\beta^2+1}, \\
&p = (e^{q\delta t}-d)/(u-d), \\
&\beta = \frac{1}{2} (e^{-q\delta t} + e^{(q+\sigma^2)\delta t}).
\end{align*}
The standard deviation $\sigma$ is also known as the volatility of the stock. For more details on deriving $u$ and $p$, see Hull \cite{Hull} or Seydel \cite{Seydel}. The above description follows the notations and development in Section 1.4 of Seydel \cite{Seydel} closely.

Starting with the current stock price in the market $S_0$, a grid of possible future stock prices $S_{ij}$ is built using $u$ and $p$. Algorithm \ref{alg:build-grid} shows the procedure to build a binomial tree of simulated future stock prices and calculate the payoffs at time $T$ for a call option, for which, $V(S_T)$ is given by \eqref{eq:call} at each $j$ at time $T$. Therefore, $V_{Nj} = \max\{S_{Nj}-K,0\}, j = 0,\ldots,N$, where $V_{ij}$ is $V(S_{ij})$. Figure \ref{fig:2bin} shows a two-step recombinant binomial tree for a call option starting at the stock price $S$ with the stock price evolution and option payoffs.

\begin{algorithm}
	\begin{algorithmic}
		\For{$i = 1,2,\ldots,N$}
		
		$S_{ij} = S_0u^jd^{i-j}$ for $j=0,1,\ldots,i$
		
		\EndFor
		\For{$j = 0, \ldots, N$}
		\State $V_{Nj} \gets \max\{S_{Nj}-K,0\}$
		\EndFor
	\end{algorithmic}
	\caption{Build the grid of stock prices and calculate option 
		payoffs for binomial method.}
	\label{alg:build-grid}
\end{algorithm}

In order to calculate the option value $V(S_0)$, the probabilities of reaching each of the leaf nodes of the tree must be calculated. These may be obtained from the probabilities of traversing each of the Bernoulli paths of dimension $N$. Since we assume that $p$ is constant from A\ref{A2}, all the paths with the same number of up and down movements have the same probability of being traversed. The option value $V(S_0)$ is computed as the expectation of the payoffs discounted to the starting time $t=0$ at the annual interest rate $q$ as
\begin{equation}
V(S_0) = e^{-qT}\sum_{i=0}^{N}p(i)V_{Ni} = e^{-qT}\sum_{i=0}^N\binom{N}{i}p^{i}{(1-p)}^{N-i} V_{Ni},
\label{eq:f8}
\end{equation}
where $p(i)=\binom{N}{i}p^{i}{(1-p)}^{N-i}$ is the probability of traversing paths ending at leaf node $i$, whose payoff 
is $V_{Ni}$.

\begin{figure}\centering
	\includegraphics[width=0.65\textwidth]{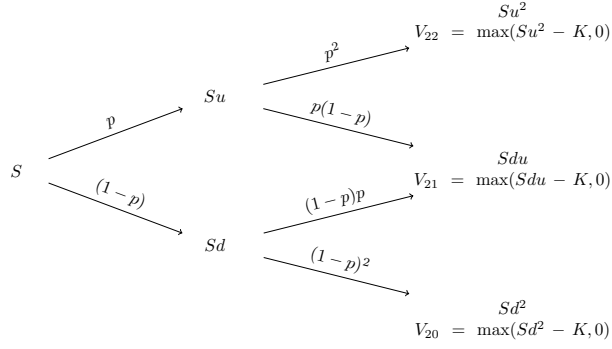}
	\caption{A two-step recombinant binomial tree with call option payoffs.}
	\label{fig:2bin}
\end{figure}

Let $\vec{X} = (X_1,\ldots,X_N)$ represent a Bernoulli path where each $X_i \sim \text{Bernoulli}(p)$ independently for $i = 1, \ldots, N$. Figure \ref{fig:2xbin} shows the two-step binomial tree in Figure \ref{fig:2bin} with Bernoulli paths to leaf nodes shown as vectors. The probability of taking path $\vec{x}$ is given by
\begin{equation*}
\Prob(\vec{X} = \vec{x}) = p^{\vec{x}'\vec{1}} (1-p)^{N-\vec{x}'\vec{1}},
\end{equation*}
where $\vec{1}$ is an $N$-dimensional vector of ones. 
Since there are ${{N}\choose{i}}$ ways of reaching the leaf node $i$, 
\begin{align}
\Prob\{ \text{reaching terminal node i} \}
&= \binom{N}{i} p^{i} (1-p)^{N-i} \nonumber \\
&= \sum_{\substack{\vec{x}\in \mathbb{B}^{N}: \vec{x}'\vec{1} = i}} p^{\vec{x}'\vec{1}} (1-p)^{N-\vec{x}'\vec{1}}.
\label{eq:f11}
\end{align}
Substituting \eqref{eq:f11} in \eqref{eq:f8}, we obtain
\begin{equation}
V(S_0) = e^{-qT} \sum_{i=0}^{N} V_{Ni}
\sum_{\substack{\vec{x}\in \mathbb{B}^{N}: \vec{x}'\vec{1} = i}}
p^{\vec{x}'\vec{1}} (1-p)^{N-\vec{x}'\vec{1}}.
\label{eq:f12}
\end{equation}
If the magnitudes and probabilities of up and down movements at each time step are constant, there is little computational advantage in evaluating the option value using \eqref{eq:f12} as opposed to \eqref{eq:f8}. However, if the tree is built using time-varying up and down movements with corresponding probability $p_t$ of an up movement at time $t$, or if the payoffs depend on the path $\vec{x}$, the model in \eqref{eq:f8} cannot be used. Let $p(\vec{x})$ be the probability of traversing the Bernoulli path $\vec{x}$ and $V_N(\vec{x})$ be the corresponding payoff. Since the space of Bernoulli paths is $\mathbb{B}^N$, \eqref{eq:f12} becomes
\begin{equation}\label{eq:f13}
V(S_0) = e^{-qT} \sum_{i=0}^{N}\sum_{\substack{\vec{x}\in \mathbb{B}^{N}: \vec{x}'\vec{1} = i}}p(\vec{x})V_N(\vec{x}) = e^{-qT}\sum_{\substack{\vec{x}\in \mathbb{B}^{N}}}V_N(\vec{x})p(\vec{x}),
\end{equation}
where $p(\vec{x}) = \prod_{i=1}^{N} p_{i}^{I(x_i = 1)} (1 - p_i)^{I(x_i = 0)}$ and $I$ is the indicator function. Note that \eqref{eq:f13} is similar to \eqref{eq:expval}. We seek to parallelize the computation of the option value $V(S_0)$ in \eqref{eq:f13} or in general, the expected value in \eqref{eq:expval}.

\begin{figure}\centering
	\includegraphics[width=0.65\textwidth]{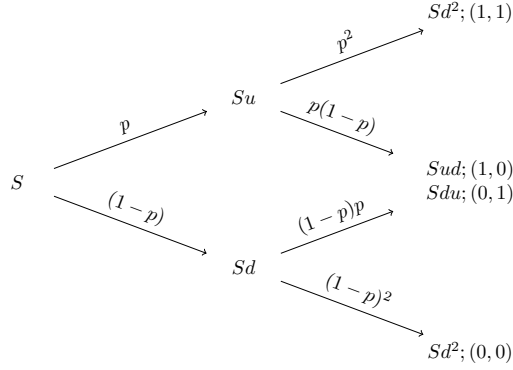}
	\caption{A two-step binomial tree with Bernoulli paths.}
	\label{fig:2xbin}
\end{figure}

\section{Parallel Bernoulli Path Algorithm}
\label{sec:parallel}
Computation of the expected value \eqref{eq:f13} quickly becomes expensive as $N$ increases, as $2^N$ Bernoulli paths must be considered. For example, taking $N = 24$ yields $16,\!777,\!216$ possible paths. The computational burden can be efficiently shared by multiple processors by noting that the problem is embarrassingly parallel. Works such as Ganesan et al. \cite{Ganesan} and Kolb and Pharr \cite{Nvidia} have proposed parallel methods for evaluating option pricing models based on backward induction in a binary tree. To the best of our knowledge, the approach to parallelize the expected value computation using Bernoulli paths has not been considered before. In our previous work, Popuri et al. \cite{Popuri2013} used a master-worker paradigm where the master process builds the tree, calculates the payoffs, allocates the terminal nodes to the worker processes, and collects the calculated values from each worker process to construct the final result. Even though the processes do not communicate with each other during the calculation, there is substantial initial communication between the master and the worker processes.

Our approach here is based on the SPMD paradigm where a single program is executed on all the processes in parallel. 
It is not necessary for one process to preallocate the workload to individual processes; instead, each process can determine its share of the $2^N$ paths to work on. This is possible using the unique rank assigned to each process. 
Each process computes a local expected value on its partition of the sample space, and the final expected value is computed by summing across all processes. This summation is accomplished in the Message Passing Interface (MPI) framework through a \emph{reduce} operation that coordinates communication between processes in an efficient way \cite{Pacheco}. 
Utilizing the unique ranks of the processes to parallelize an algorithm 
is a common theme in parallel computing \cite{Pacheco}. For example, Swarztrauber and Sweet \cite{Swarztrauber1989} use 
the binary representation of the data to map computations to processes in a parallel direct solution of Poisson's equation. 
Here, we use the leading bits of the Bernoulli paths to map the paths onto processes.

Suppose there are $M$ parallel processes with ranks $m = 0, \ldots, M-1$; note that ranks traditionally start at 0 in the MPI framework. The process with rank $m$ will be referred to as ``process $m$''. We assume that $M \leq N$ and that $M$ is a power of $2$. Let $r = \log_2(M)$ so that the rank $m$ of a process can be written with the $r$-digit binary representation $m = z_{r-1} 2^{r-1} + \cdots + z_1 2^1 + z_0 2^0$, where each $z_j \in \mathbb{B}$. Process $m$ is assigned all paths $\vec{x}$ with prefix $(z_{r-1}, \ldots, z_1, z_0)$; this set of $2^{N-r}$ paths is denoted 
\begin{align*}
\mathbb{B}_m^{N} = \{ \vec{x} \in \mathbb{B}^N : x_1 = z_{r-1}, \ldots, x_{r-1} = z_1, x_r = z_0 \}.
\end{align*}
Note that the sets $\mathbb{B}_m^N$ form a partition of $\mathbb{B}^N$. Figure \ref{fig:algo} shows a diagram of the mapping from $m$ to $\mathbb{B}_m^{N}$. For each Bernoulli path $\vec{x} \in \mathbb{B}_m^{N}$, process $m$ computes the probability of traversing the path $p(\vec{x})$ and the payoff value $V_N(\vec{x})$. The local expected value of the payoffs $V_m$ on process $m$ is calculated as
\begin{equation}
\label{eq:f14}
V_m = e^{-qT}\sum_{\substack{\vec{x} \in \mathbb{B}_m^{N}}} p(\vec{x}) V_N(\vec{x}).
\end{equation}
Finally, local expected values are summed to produce the final result
\begin{equation}
\label{eq:f15}
V(S_0) = \sum_{m=0}^{M-1}V_m;
\end{equation}
this is implemented by an MPI reduce operation to obtain the result on process $0$. The computation of \eqref{eq:f15} requires $2^{N-r}$ steps on $2^r$ parallel processes rather than $2^N$ steps on a single process, as required in the serial computation. The method can be extended to the case when $M$ is not exactly a power of 2 if we are willing to forfeit perfect load balancing. For example, we can consider a partition $\mathbb{B}^N = \mathbb{B}_0^N \cup \cdots \cup \mathbb{B}_K^N$ for some $K >> M$. Process $0$ can handle $\mathbb{B}_m^N$ for $m = 0, M, 2M, \ldots$, process $1$ can handle $\mathbb{B}_m^N$ for $m = 1, M+1, 2M+1, \ldots$, and so forth. Note that the idea of partitioning the paths is applicable to trees with more than two branches at each node. Examples of such trees include trinomial trees\cite{Boyle1986}.

\begin{figure} \centering
	\includegraphics[width=0.65\textwidth]{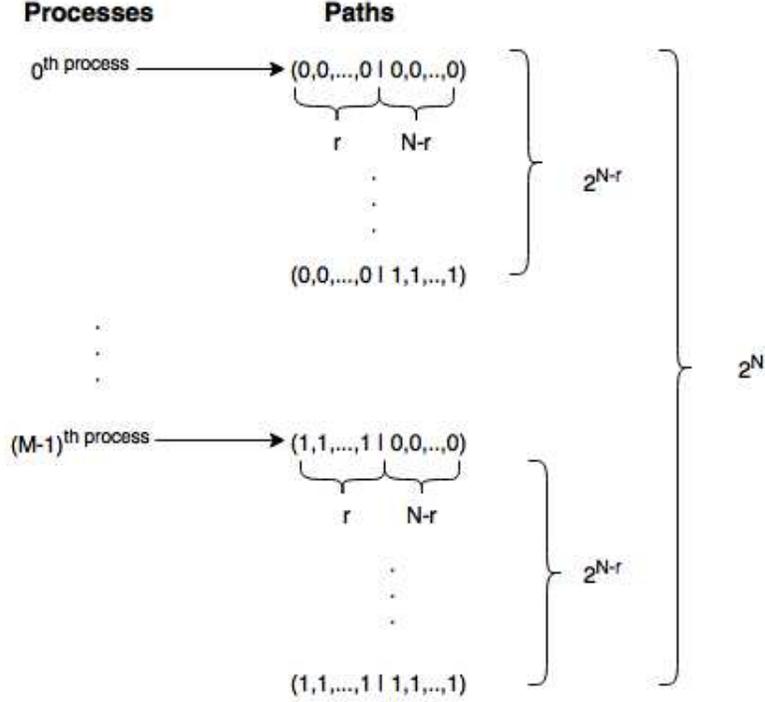}
	\caption{Process-Bernoulli path mapping}
	\label{fig:algo}
\end{figure}


\section{Monte Carlo Estimation and Variance Reduction}
\label{sec:montecarlo}
Recall that the number of Bernoulli paths in the set $\mathbb{B}^N$ grows exponentially with $N$. When $N$ becomes large, it is infeasible to compute the expected value \eqref{eq:f13} exactly, even with a reasonably large number of processors. Monte Carlo (MC) estimation provides a way to approximate a complicated expected value without enumerating the entire sample space. In this section, we discuss an MC method that uses the partitioning scheme from Section \ref{sec:parallel} to approximate the result using $M$ parallel processes. The $m^{\text{th}}$ process is given the responsibility of drawing from $\mathbb{B}_m^N$, for $m = 0, \ldots, M-1$, so that we effectively enumerate the first $r = \log_2 M$ steps of each path, and draw the rest through Monte Carlo. This provides a reduction in variance over a basic MC estimator that uses the same number of draws. 

Define
\begin{equation*}
\theta = \E[V(\vec{X})] = \sum_{\vec{x} \in \mathbb{B}^N} V(\vec{x})p(\vec{x}),
\end{equation*}
where the suffix $N$ in $V(\vec{x})$ is dropped for notational convenience. The option value in
\eqref{eq:f13} can be written as $V = e^{-qT}\theta$. Given an estimate $\hat{\theta}$ of $\theta$, an estimate of $V$ is $\hat{V} = e^{-qT} \hat{\theta}$, its variance is $\Var(\hat{V}) = e^{-2qT} \Var(\hat{\theta})$, and an estimate of the variance is $\widehat{\Var}(\hat{V}) = e^{-2qT} \widehat{\Var}(\hat{\theta})$. Therefore, we will focus on estimating $\theta$ for the remainder of this section.

Let $\vec{x}_1,\ldots,\vec{x}_R$ be $R$ independent and identically distributed (i.i.d.) Bernoulli paths sampled from $\mathbb{B}^{N}$. Then the MC estimator of $\theta$ is given by
\begin{equation*}
\hat{\theta} = \frac{1}{R}\sum_{i=1}^{R}V(\vec{x}_i)
\end{equation*}
and its variance is
\begin{equation}
\label{eq:mc2}
\Var(\hat{\theta}) = \frac{1}{R}\Var\lbrack V(\vec{X})\rbrack,
\end{equation}
which can be estimated from the MC draws by
\begin{equation*}
\widehat{\Var}(\hat{\theta}) = \frac{1}{R^2} \sum_{i=1}^R \Big( V(\vec{x}_i) - \hat{\theta} \Big)^2.
\end{equation*}
In Section \ref{sec:parallel}, we partitioned the space $\mathbb{B}^N$ of Bernoulli paths into $\mathbb{B}_0^N, \ldots, \mathbb{B}_M^N$. Let $\mathcal{D}_m$ denote the event $\lbrack \vec{X} \in \mathbb{B}_m^N\rbrack$ which occurs with probability $\Prob(\mathcal{D}_m)$, for $m=0, \ldots, M-1$. Furthermore, consider the partitioning $\vec{X} = (\vec{Z}, \vec{Y})$ where $\vec{Z}\in \mathbb{B}^r$ and $\vec{Y}\in \mathbb{B}^{N-r}$. 
We can now write $\theta$ as
\begin{equation}
\label{eq:mc9}
\theta = \sum_{m=0}^{M-1}\E\lbrack V(\vec{X})\mid\mathcal{D}_m\rbrack \Prob(\mathcal{D}_m),
\end{equation}
where
\begin{align*}
\Prob(\mathcal{D}_m)
= \sum_{\vec{y}\in \mathbb{B}^{N-r}}\Prob(\vec{Z}=\vec{z}_m,\vec{Y}=\vec{y})
= \Prob(\vec{Z}=\vec{z}_m)
\end{align*}
and $\vec{z}_m$ is the binary representation of $m$ corresponding to the rank of the $m^{\text{th}}$ process. Let $\theta^{(m)} = \E\lbrack V(\vec{X})\mid\mathcal{D}_m\rbrack$ and let $\vec{x}_1^{(m)}, \ldots, \vec{x}_{R_m}^{(m)}$ be an i.i.d.~sample from the distribution of paths on $\mathbb{B}_{m}^{N}$ for each $m=0, \ldots, M-1$. Suppose $\sum_{m=0}^{M-1} R_m = R$ so that the sample size used is as in the basic MC estimator. The estimator
\begin{equation*}
\hat{\theta}_s^{(m)} = \frac{1}{R_m}\sum_{i=1}^{R_m} V(\vec{x}_i^{(m)})
\end{equation*}
is an unbiased estimator of $\theta^{(m)}$ with variance $\frac{1}{R_m}\Var\lbrack V(\vec{X})\mid\mathcal{D}_m\rbrack$. Substituting $\hat{\theta}_s^{(m)}$ for $\theta^{(m)}$ in \eqref{eq:mc9} yields the Partitioned MC estimator
\begin{equation}
\label{eq:mc12}
\hat{\theta}_s = \sum_{m=0}^{M-1}\hat{\theta}_s^{(m)} \Prob(\mathcal{D}_m).
\end{equation}
Following Rubinstein and Kroese \cite{Rubinstein}, we choose sample sizes $R_m$ proportional to $\Prob(\mathcal{D}_m)$ as $R_m = R \cdot \Prob(\mathcal{D}_m)$ for each $m$. With this choice, and ignoring that $R \cdot \Prob(\mathcal{D}_m)$ likely will not be an exact integer, the variance of the Partitioned MC estimator can be written
\begin{align}
\Var(\hat{\theta}_s) &= \sum_{m=0}^{M-1}\Var(\hat{\theta}^{(m)})\lbrack \Prob(\mathcal{D}_m)\rbrack^2 \nonumber \\
&= \frac{1}{R}\sum_{m=0}^{M-1}\Var\lbrack V(\vec{X})\mid\mathcal{D}_m\rbrack \Prob(\mathcal{D}_m).
\label{eq:mc13}
\end{align}
A corresponding variance estimator is
\begin{align*}
\widehat{\Var}(\hat{\theta}_s)
&= \frac{1}{R}\sum_{m=0}^{M-1}
\left[ \frac{1}{R_m} \sum_{i=1}^{R_m} \Big( V(\vec{x}_i^{(m)}) - \hat{\theta}_m \Big)^2 \right]
\Prob(\mathcal{D}_m) \nonumber \\
&= \frac{1}{R^2}\sum_{m=0}^{M-1}
\sum_{i=1}^{R_m} \Big( V(\vec{x}_i^{(m)}) - \hat{\theta}_m \Big)^2.
\end{align*}
To verify that $\hat{\theta}_s$ gives a variance reduction over $\hat{\theta}$, the law of total variation gives
\begin{align}
\Var\lbrack V(\vec{X})\rbrack &= \E_{\mathcal{D}}\Var\lbrack V(\vec{X})\mid\mathcal{D}\rbrack + \Var_{\mathcal{D}}\E\lbrack V(\vec{X})\mid\mathcal{D}\rbrack \\
&= \sum_{m=0}^{M-1}\Var\lbrack V(\vec{X})\mid\mathcal{D}_m\rbrack \Prob(\mathcal{D}_m) + \Var_{\mathcal{D}}\E\lbrack V(\vec{X})\mid\mathcal{D}\rbrack \nonumber\\
&= R\Var(\hat{\theta}_s)+\Var_{\mathcal{D}}\E\lbrack V(\vec{X})\mid\mathcal{D}\rbrack,
\label{eq:mc14}
\end{align}
where the last equality is from \eqref{eq:mc13}. Substituting the left hand side in \eqref{eq:mc14} in terms of $\Var(\hat{\theta})$ from \eqref{eq:mc2} and dividing both sides by $R$ we get
\begin{equation}
\label{eq:mc15}
\Var(\hat{\theta}) = \Var(\hat{\theta}_s) + \frac{\Var_{\mathcal{D}}\E\lbrack V(\vec{X})\mid\mathcal{D}\rbrack}{R}.
\end{equation}
Note that $\Var_{\mathcal{D}}\E\lbrack V(\vec{X})\mid\mathcal{D}\rbrack = 0$ if $V(\vec{X})\mid\mathcal{D}$ does not depend on the first $r$ steps of the Bernoulli paths or when $r \in \{ 0, N \}$, that is, when the number of processes $M \in \{ 1, 2^N \}$. When $M=1$, the Partitioned MC method is same as the basic MC method and when $M=N$, it is same as the exact expected value in \eqref{eq:f13}. Since the payoff $V(\vec{X})$ is assumed to depend on the entire path $\vec{X}$, the second term in the right hand side of \eqref{eq:mc15} is greater than $0$ when $0<r<N$ and, therefore, the Partitioned MC estimator in \eqref{eq:mc12} typically yields strict reduction in the variance. The reduction will be more pronounced when $V(\vec{X})$ are heterogeneous across the $\mathcal{D}_m$ and homogeneous within each $\mathcal{D}_m$.

\begin{remark}
	An interesting variation of the Partitioned MC method is to reuse the same sample of $R$ draws from $\mathbb{B}^{N-r}$ on all processes. Consider again the partitioning $\vec{X} = (\vec{Z}, \vec{Y})$, and suppose $\vec{Z}$ and $\vec{Y}$ are independent. Let $\vec{x}_i^{(m)} = (\vec{z}_{m},\vec{y}_{i})$, where $\vec{y}_1, \ldots, \vec{y}_R$ are i.i.d.~draws from the distribution of $\vec{Y}$ and $\vec{z}^{(m)} = (z_{r-1}^{(m)}, \ldots, z_1^{(m)}, z_0^{(m)})$ is the binary representation of $m$ corresponding to the rank of the $m^{\text{th}}$ process. Then an estimate for $\theta^{(m)}$ is given by
	\begin{align}
	\tilde{\theta}^{(m)} = \frac{1}{R}\sum_{i=1}^{R}V(\vec{x}_i^{(m)} )
	\end{align}
	and $\theta$ in \eqref{eq:mc9} can be estimated unbiasedly by 
	\begin{align}
	\tilde{\theta} &= \sum_{m=0}^{M-1} \tilde{\theta}^{(m)} \Prob(\mathcal{D}_m) \nonumber \\
	&= \frac{1}{R} \sum_{i=1}^{R} \sum_{m=0}^{M-1} V(\vec{x}_i^{(m)})\Prob(\mathcal{D}_m) \nonumber \\
	&= \frac{1}{R} \sum_{i=1}^{R} \E_{\vec{X}\mid\vec{Y}} \lbrack V(\vec{Z},\vec{y}_i)\rbrack.
	\label{eq:mc16}
	\end{align}
	The variance of this estimator is
	$\Var(\tilde{\theta}) = \frac{1}{R} \Var_{\vec{Y}} \E_{\vec{X} \mid \vec{Y}} \lbrack V(\vec{X})\rbrack$ and an estimate of the variance from the MC sample is
	\begin{align*}
	\widehat{\Var}(\tilde{\theta})
	&= \frac{1}{R^2} \sum_{i=1}^{R} \left(\E_{\vec{X}\mid\vec{Y}} \lbrack V(\vec{Z},\vec{y}_i)] - \tilde{\theta}^{(m)} \right)^2 \\
	&= \frac{1}{R^2} \sum_{i=1}^{R} \left( \sum_{m=0}^{M-1} V(\vec{x}_i^{(m)})\Prob(\mathcal{D}_m) - \tilde{\theta}^{(m)} \right)^2.
	\end{align*}
	We refer to $\tilde{\theta}$ as the Shared Sample MC estimator. Now, the variance of $V(\vec{X})$ can be written as 
	\begin{align}
	\Var\lbrack V(\vec{X})\rbrack
	&= \Var_{\vec{Y}}\E\lbrack V(\vec{X})\mid\vec{Y}\rbrack + \E_{\vec{Y}}\Var\lbrack V(\vec{X})\mid\vec{Y}\rbrack \nonumber \\
	&= R\Var(\tilde{\theta})+\E_{\vec{Y}}\Var\lbrack V(\vec{X})\mid\vec{Y}\rbrack.
	\label{eq:mc17}
	\end{align}
	Substituting the left hand side in \eqref{eq:mc17} in terms of $\Var(\hat{\theta})$ from \eqref{eq:mc2} and dividing both sides by $R$ we get
	\begin{equation}
	\label{eq:mc18}
	\Var(\hat{\theta}) = \Var(\tilde{\theta}) + \frac{\E_{\vec{Y}}\Var\lbrack V(\vec{X})\mid\vec{Y}\rbrack}{R}.
	\end{equation}
	Again, the second term in the right hand side of the \eqref{eq:mc18} is zero if and only if $V(\vec{X})\mid \vec{Y}$ does not depend on the first $r$ steps of the Bernoulli paths, for $0 < r < N$. Since we assume $V(\vec{X})$ depends on the entire path $\vec{X}$, $\Var(\tilde{\theta})$ is strictly less than $\Var(\hat{\theta})$. The following result summarizes the relationship among the variances of $\hat{\theta}$, $\hat{\theta}_s$, and $\tilde{\theta}$.
	
	\begin{theorem}
		\label{result:strathybrid}
		Suppose $R_m = R$ for $m=0, \ldots, M-1$, $0<r<N$, and $V(\vec{Z}_k,\vec{Y})$ and $V(\vec{Z}_l,\vec{Y})$ are positively correlated for all 
		$k,l \in \{ 0, \ldots, M-1 \}$. Then $\Var(\hat{\theta}_s) \leq \Var(\tilde{\theta}) \leq \Var(\hat{\theta})$.
	\end{theorem}
	
	\begin{proof}
		We have already shown that $\Var(\tilde{\theta}) \leq \Var(\hat{\theta})$. Now we have $\Var(\hat{\theta}_s) = \frac{1}{R}\sum_{m=0}^{M-1}\Var\lbrack V(\vec{X})\mid\mathcal{D}_m\rbrack \lbrack \Prob(\mathcal{D}_m)\rbrack^2$, and from \eqref{eq:mc16},
		\begin{align*}
		\Var(\tilde{\theta}) &= \frac{1}{R}\Var_{\vec{Y}}\left[ \sum_{m=0}^{M-1}V(\vec{Z}_m,\vec{Y})\Prob(\mathcal{D}_m) \right] \nonumber \\
		&= \frac{1}{R} \sum_{m=0}^{M-1}\Var_{\vec{Y}}\lbrack V(\vec{Z}_m,\vec{Y})\rbrack \lbrack \Prob(\mathcal{D}_m) \rbrack^2 + 
		\frac{1}{R} \mathop{\sum \sum}_{k \neq l} \Cov\Big( V(\vec{Z}_k,\vec{Y}), V(\vec{Z}_l,\vec{Y}) \Big) \Prob(\mathcal{D}_k)\Prob(\mathcal{D}_l) \label{thm:eq1}\\
		&= \Var(\hat{\theta}_s)+\frac{1}{R} \mathop{\sum \sum}_{k \neq l} \Cov\Big( V(\vec{Z}_k,\vec{Y}), V(\vec{Z}_l,\vec{Y}) \Big)
		\Prob(\mathcal{D}_k)\Prob(\mathcal{D}_l) \nonumber \\
		&\geq \Var(\hat{\theta}_s).
		\end{align*}
	\end{proof}
\end{remark}

\section{Application to option pricing}
\label{sec:results}
We implemented the method described in Section \ref{sec:parallel} to value a put option using the \proglang{R} 3.2.2 and \proglang{Julia} 0.4.6 programming environments. Computations were run on a distributed cluster with compute nodes, each having two Intel E5-2650v2 Ivy Bridge (2.6 GHz, 20 MB cache) processors with 8 cores per node, for a total of 16 cores per node. All nodes have 64 GB of main memory and are connected by a quad-data rate InfiniBand interconnect. Open MPI 1.8.5 \cite{GabrielEtAl2004} was used as the underlying implementation of the MPI framework.

\proglang{R} is a statistical computing environment that facilitates advanced data analysis, provides graphical capabilities and an interpreted high-level programming language \cite{Rwebsite}. On top of the statistical, computational, and programmatic features available in the core \proglang{R} environment, additional capabilities are available through numerous packages which have been contributed by the user community. The \pkg{Rmpi} \cite{Yu2002} and \pkg{pbdMPI} \cite{pbdR2012} packages may be used to write MPI programs from \proglang{R}. Results shown in this section are based on \pkg{Rmpi}, but \pkg{pbdMPI} performed similarly in our experience. The package \pkg{Rcpp} \cite{Rcpp} facilitates integration of \proglang{C++} code into \proglang{R} programs, which can substantially improve performance at the cost of an increased programming burden. We have not yet explored \pkg{Rcpp} in our implementation, but note its potential use.

\proglang{Julia} is a recently developed programming language that is gaining popularity in scientific computing, data analysis, and high performance computing \cite{Julialang}. It is a compiled language that uses the Low Level Virtual Machine Just-in-Time technology \cite{LLVM} to generate an optimized version of the source code compiled to the machine level. \proglang{Julia} provides a number of computational and statistical capabilities, both in the core environment and through packages contributed by the user community. We have used the package \pkg{MPI} \cite{MPIJulia} to run MPI programs in \proglang{Julia}. Integration with \proglang{C++} is also possible in \proglang{Julia} through packages such as \pkg{CxxWrap} and \pkg{Cpp}, but we have not yet explored their use. Our implementation uses native \proglang{Julia} code with the \pkg{MPI} package. Because \proglang{Julia} is compiled into machine-level code, it is expected that a program written in \proglang{Julia} will perform better than an equivalent program written in \proglang{R}. Performance results later in this section confirm our hunch.

Listing \ref{lst:jalgo} shows a snippet of our \proglang{Julia} implementation of the parallelization method. Since the structure of our \proglang{R} and \proglang{Julia} implementations are similar, we do not show a similar listing of our \proglang{R} code. In line 1 we load the \pkg{MPI} package. Since our implementation follows the SPMD paradigm, the same code runs on all the processes. The rank of the process on which the code is being run is requested on line 5 and on line 6 the total number of processes in the MPI communicator is requested. As the while loop at line 11 shows, each process works on $2^{N-r}$ out of the total $2^N$ Bernoulli paths. Note the construction of the full Bernoulli path in line 13 by prepending the binary representation of the rank of the specific process on which the code is being run to the current $(N-r)$-dimensional Bernoulli path. The function call to \code{calc\_path\_prob} on line 14 calculates the probability of traversing the Bernoulli path constructed in line 13. The function call to \code{calc\_payoff} on line 15 calculates the option payoff; their function definitions are not shown because they are independent of the parallelization method. Finally, on line 21, expected values from individual processes are summed together to obtain the final answer at process 0. 

\begin{lstfloat}
\centering
\begin{lstlisting}
import MPI
.
MPI.Init()
comm = MPI.COMM_WORLD
id = MPI.Comm_rank(comm)
M = MPI.Comm_size(comm)
.
r = log2(M)
l_n = convert(Int64, 2^(N-r))
.
while i < l_n
    node = i
    path = cat(2, integer_base_b(id, 2, r), integer_base_b(node, 2, N-r))
    p_vt = calc_path_prob(path, probs)
    vt = calc_payoff(S, K, u, d, opt_type, path)
    v += p_vt*vt
    i += 1
end
.
v = exp(-q*T)*v
reduced_v = MPI.Reduce(v, MPI.SUM, 0, comm)
\end{lstlisting}
\caption{A \proglang{Julia} implementation of the parallel Bernoulli path algorithm.}
\label{lst:jalgo}
\end{lstfloat} 

We take a put option as an example to illustrate our methodology. We set a strike price of $K=10$. Current price and volatility of the asset are $S=5$ and $\sigma=0.30$, respectively. Risk-free interest rate is $q=6\%$ and time to maturity $T$ is one year. Tables \ref{tab:r_perf}(a) and \ref{tab:julia_perf}(a) show the wall clock runtimes of our \proglang{R} and \proglang{Julia} implementations, respectively, for problem sizes $N=16, 20, 24, 28$, and $32$. While both the implementations scale well with the number of processes $M$, the \proglang{Julia} implementation is roughly $10$ times faster than \proglang{R}. Our \proglang{R} program for $N=32$ on a single process ($M=1$) resulted in an overflow in the loop that computes the expected value since $2^{31}-1$ is the maximum integer value that can be stored in \proglang{R}. As a result, the runtime for this particular case is recorded as N/A in Table \ref{tab:r_perf}(a). If $T_M$ is the runtime taken for $M$ number of processes, the speedup $S_M$ and efficiency $E_M$ for $M$ are defined as $T_1 / T_M$ and $S_M / M$ respectively. If the program scales up perfectly to $M$ processes, ideal values $S_M=M$ and $E_M=1$ are obtained. These numbers indicate the scalability of the program. Since our \proglang{R} program did not run on a single process for $N=32$, we take the speedup for this case to be $2 \cdot T_2 / T_M$, $M = 2,\ldots,64$ and for $M=1$ and $M=2$, the speedups are taken to be $1$ and $2$ respectively. Tables \ref{tab:r_perf}(b) and \ref{tab:julia_perf}(b) show the speedups and Tables \ref{tab:r_perf}(c) and \ref{tab:julia_perf}(c) show the efficiency numbers of our \proglang{R} and \proglang{Julia} implementations, respectively. The plots in Figure \ref{fig:speedupR} visualize the speedup and efficiency numbers in Tables \ref{tab:r_perf}(b) and \ref{tab:r_perf}(c), respectively, and Figure \ref{fig:speedupJulia} shows the corresponding plots from Table \ref{tab:julia_perf}. These plots visually confirm our conjecture that \proglang{Julia} is more efficient than \proglang{R} for our problem. Note that for a fixed problem size, there is a reduced advantage in the speedup beyond a certain number of tasks. This is because the overhead of coordinating the tasks begins to dominate the time spent doing useful calculations; see Pacheco \cite{Pacheco} for more details. This can be seen for $N=16$ in Table \ref{tab:r_perf}(a) and for $N=16$ and $N=20$ in Table \ref{tab:julia_perf}(a). In both the cases, since the very small total runtime is mostly dominated by the near constant time consumed by the 
MPI reduce operation, the speedup and efficiency numbers are significantly lower than larger sized problems. 
The relatively high speedup and efficiency numbers for $N=32$ in Tables \ref{tab:r_perf}(b) and \ref{tab:r_perf}(c) 
resulted from fixing the speedup for $M=2$ to be $2$ and calculating the rest of the speedup and efficiency 
numbers relative to $M=2$. Since the runtimes for \proglang{Julia} are roughly $10$ times faster than those 
for \proglang{R}, if we estimate the runtime for $N=32$ and $M=1$ for \proglang{R} and accordingly calculate the speedup 
and efficiency numbers, we will notice that the numbers drop and are comparable to rest of the cases. 

We implemented the Monte Carlo estimation methods described in Section \ref{sec:montecarlo} for Asian and Lookback options, which are both path-dependent \cite{Hull}. In an Asian option, the asset price $S_T$ at the time of maturity is replaced in the option payoff function with the arithmetic average of $\{ S_t : t=1, \ldots, N \}$. Therefore, in the binomial tree model, the payoff for an Asian put option is given by
\begin{equation}
V(\vec{x}) = \max\{K-S^*,0\},
\end{equation}
where $S^* = \frac{1}{N}\sum_{t=1}^{N}{S_t(\vec{x})}$, $S_t(\vec{x})$ is the asset value at time $t$ followed on the Bernoulli path $\vec{x}$. In a Lookback option, either the strike price $K$ or the asset price $S_T$ at the time of maturity are replaced in the payoff function by the maximum or minimum of $\{S_t\}$ respectively. Here we consider a Fixed Lookback put option, whose payoff is given by
\begin{equation*}
V(\vec{x}) = \max\{K-S^*,0\},
\end{equation*}
where $S^* = \min\{ S_t(\vec{x}) : t=1, \ldots, N \}$. We implemented the basic MC estimate given in section \ref{sec:montecarlo} for the Asian and Fixed Lookback put options using the binomial tree model with size $N$ to study the convergence of the estimates to the exact expected value \eqref{eq:f13}. We further implemented Partitioned and Shared Sample MC from section \ref{sec:montecarlo} to study the variance reduction property. Table \ref{tab:mc_conv} shows basic MC estimates and corresponding variance estimates of an Asian put option and a Fixed Lookback put option with parameters $K=100$, $S=20$, $q=6\%$, $\sigma=3.0$, and $T=1$, using the binomial model with tree size $N=32$. Option values calculated by exact enumeration were 82.115 for the Asian put options and 93.196 for the Fixed Lookback put option. The sample size used for the MC estimation is increased from $2^9$ to $2^{16}$, which is less than $0.01\%$ of the total number of paths. As can be seen from Table \ref{tab:mc_conv}, MC estimates for both the options converge to their respective exact values. Also, as expected, the variance estimates decrease with increasing sample size $R$. Table \ref{tab:var_reduction_mc} shows the Partitioned MC estimates $\hat{V}_s$ for both the Asian and Fixed Lookback put options and corresponding variance estimates, using a total sample size of $R=1024$ and varying the number of processes between $1$ to $64$. Note that as the number of processes increase, the sample size per process $R_m$ decreases. The estimates shown in Table \ref{tab:var_reduction_mc} are averaged over $1000$ repetitions. As expected, Table \ref{tab:mc_conv} shows that variance estimates of the Partitioned MC estimator are mostly smaller than the corresponding basic MC estimator for $R=1024$. Table \ref{tab:smc_hmc} shows the comparison of variance estimates between the Partitioned and Shared Sample MC estimates for the Asian put option with $N=32$, $R = R_m = 1024$, and $m=0, \ldots, M-1$. Shared Sample MC estimates $\tilde{V}$ and corresponding variance estimates were calculated using the expressions given in section \ref{sec:montecarlo}. Again, the estimates in Table \ref{tab:smc_hmc} are averaged over $1000$ repetitions. The results show that if $R=R_m$, and $m=0, \ldots, M-1$, the Partitioned MC method reduces the variance of the estimator more than the Shared Sample method does, as expected from Theorem \ref{result:strathybrid}. The condition on the covariance between $V(\vec{z}_k,\vec{y})$ and $V(\vec{z}_l,\vec{y})$ for all $k,l \in \{ 0, \ldots, M-1 \}$ where $k\neq l$ is satisfied for the options considered here.

\begin{table}
	\centering
	\caption{Runtime for different number of time steps for \proglang{R} implementation. For $M=1$, $N=32$, since our program failed to run because of integer overflow, runtime is shown as N/A.}
	\label{tab:r_perf}
	\begin{tabular}{rrrrrrrrrrrrrrr}
		\hline
		\multicolumn{8}{l}{(a) Wall clock time in HH:MM:SS} \\
		\multicolumn{1}{r}{\makebox[0.10in][r]{$N$}} &
		\multicolumn{1}{r}{\makebox[0.01in][r]{$M = 1 $}} &
		\multicolumn{1}{r}{\makebox[0.01in][r]{$    2 $}} &
		\multicolumn{1}{r}{\makebox[0.01in][r]{$    4 $}} &
		\multicolumn{1}{r}{\makebox[0.01in][r]{$    8 $}} &
		\multicolumn{1}{r}{\makebox[0.01in][r]{$   16 $}} &
		\multicolumn{1}{r}{\makebox[0.01in][r]{$   32 $}} &
		\multicolumn{1}{r}{\makebox[0.01in][r]{$   64 $}} \\
		\hline
		16&  00:00:04&	00:00:02&	00:00:01&	$<$00:00:01&	$<$00:00:01&	$<$00:00:01&	$<$00:00:01\\
		20&  00:01:09&	00:00:43&	00:00:23&	00:00:12&	00:00:06&	00:00:03&	00:00:02\\
		24&  00:20:22&	00:12:41&	00:06:56&	00:03:35&	00:01:47&	00:00:54&	00:00:27\\
		28&  06:00:09&	03:41:26&	01:59:31&	01:02:16&	00:31:18&	00:16:02&	00:08:00\\
		32&  N/A     &	65:54:36&	35:18:59&	18:41:10&	10:38:06&	04:44:09 &	02:22:59\\
		\hline
		\hline
		\multicolumn{8}{l}{(b) Observed speedup $S_M$} \\
		\multicolumn{1}{r}{\makebox[0.10in][r]{$N$}} &
		\multicolumn{1}{r}{\makebox[0.01in][r]{$M = 1 $}} &
		\multicolumn{1}{r}{\makebox[0.01in][r]{$    2 $}} &
		\multicolumn{1}{r}{\makebox[0.01in][r]{$    4 $}} &
		\multicolumn{1}{r}{\makebox[0.01in][r]{$    8 $}} &
		\multicolumn{1}{r}{\makebox[0.01in][r]{$   16 $}} &
		\multicolumn{1}{r}{\makebox[0.01in][r]{$   32 $}} &
		\multicolumn{1}{r}{\makebox[0.01in][r]{$   64 $}} \\
		\hline
		16&  1.00&	1.60&	2.97&	5.55&	10.88&	17.69&	27.11\\
		20&  1.00&	1.59&	2.98&	5.72&	11.62&	22.86&	44.33\\
		24&  1.00&	1.60&	2.93&	5.68&	11.38&	22.62&	44.32\\
		28&  1.00&	1.62&	3.01&	5.78&	11.50&	22.45&	45.01\\
		32&  N/A &	2.00&	3.73&	7.05&	13.99&	28.21&	55.31\\
		\hline
		\hline
		\multicolumn{8}{l}{(c) Observed efficiency $E_M$} \\
		\multicolumn{1}{r}{\makebox[0.10in][r]{$N$}} &
		\multicolumn{1}{r}{\makebox[0.01in][r]{$M = 1 $}} &
		\multicolumn{1}{r}{\makebox[0.01in][r]{$    2 $}} &
		\multicolumn{1}{r}{\makebox[0.01in][r]{$    4 $}} &
		\multicolumn{1}{r}{\makebox[0.01in][r]{$    8 $}} &
		\multicolumn{1}{r}{\makebox[0.01in][r]{$   16 $}} &
		\multicolumn{1}{r}{\makebox[0.01in][r]{$   32 $}} &
		\multicolumn{1}{r}{\makebox[0.01in][r]{$   64 $}} \\
		\hline
		16&  1.00&	0.80&	0.74&	0.69&	0.68&	0.55&	0.42\\
		20&  1.00&	0.80&	0.75&	0.71&	0.72&	0.71&	0.69\\
		24&  1.00&	0.80&	0.73&	0.71&	0.71&	0.71&	0.69\\
		28&  1.00&	0.81&	0.75&	0.73&	0.72&	0.70&	0.70\\
		32&  N/A &	1.00&	0.93&	0.88&	0.87&	0.88&	0.86\\
		\hline
	\end{tabular}
\end{table}

\begin{figure}
	\centering
	\begin{tabular}{cc}
		\includegraphics[width=0.5\textwidth]{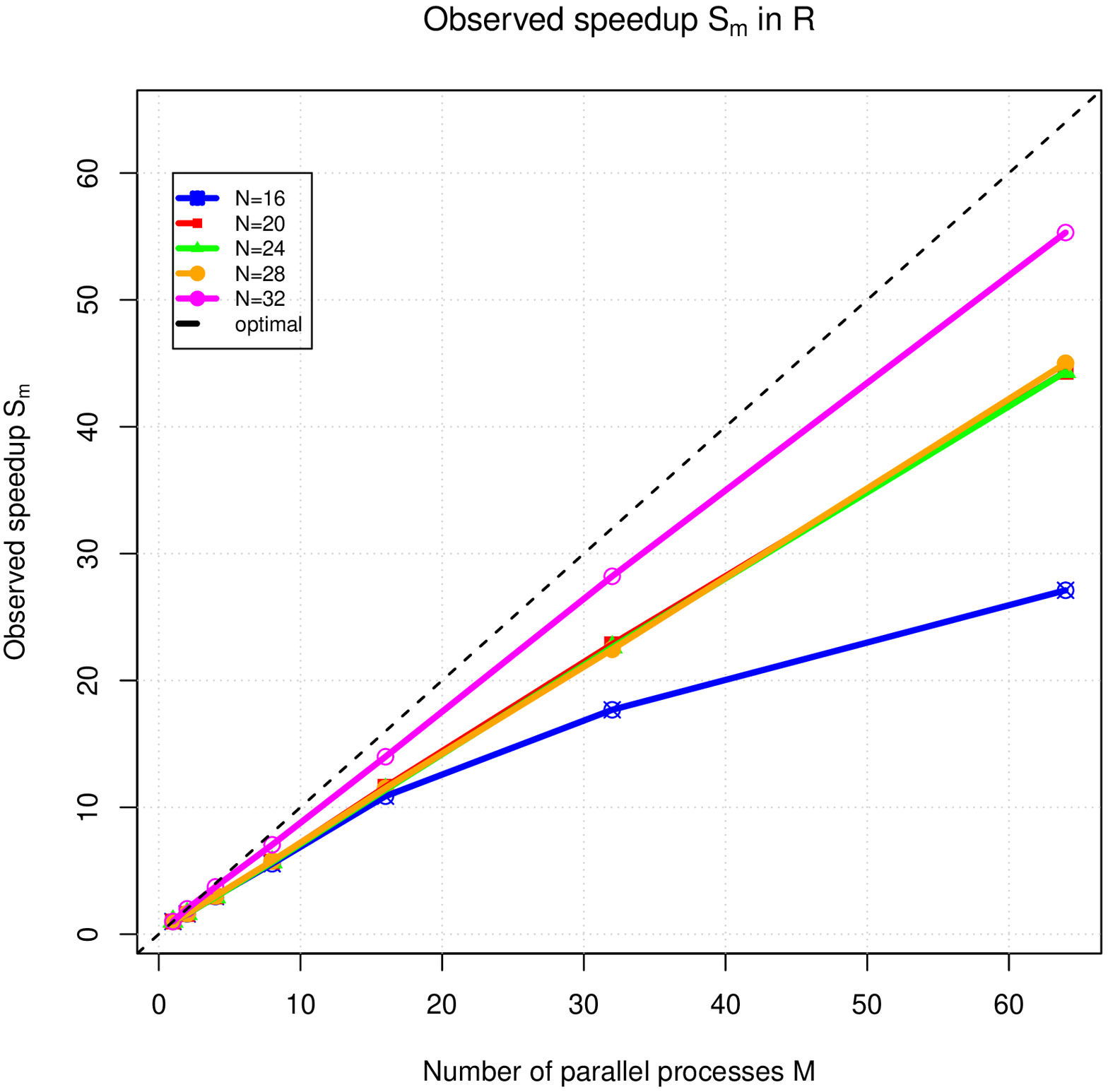} &
		\includegraphics[width=0.5\textwidth]{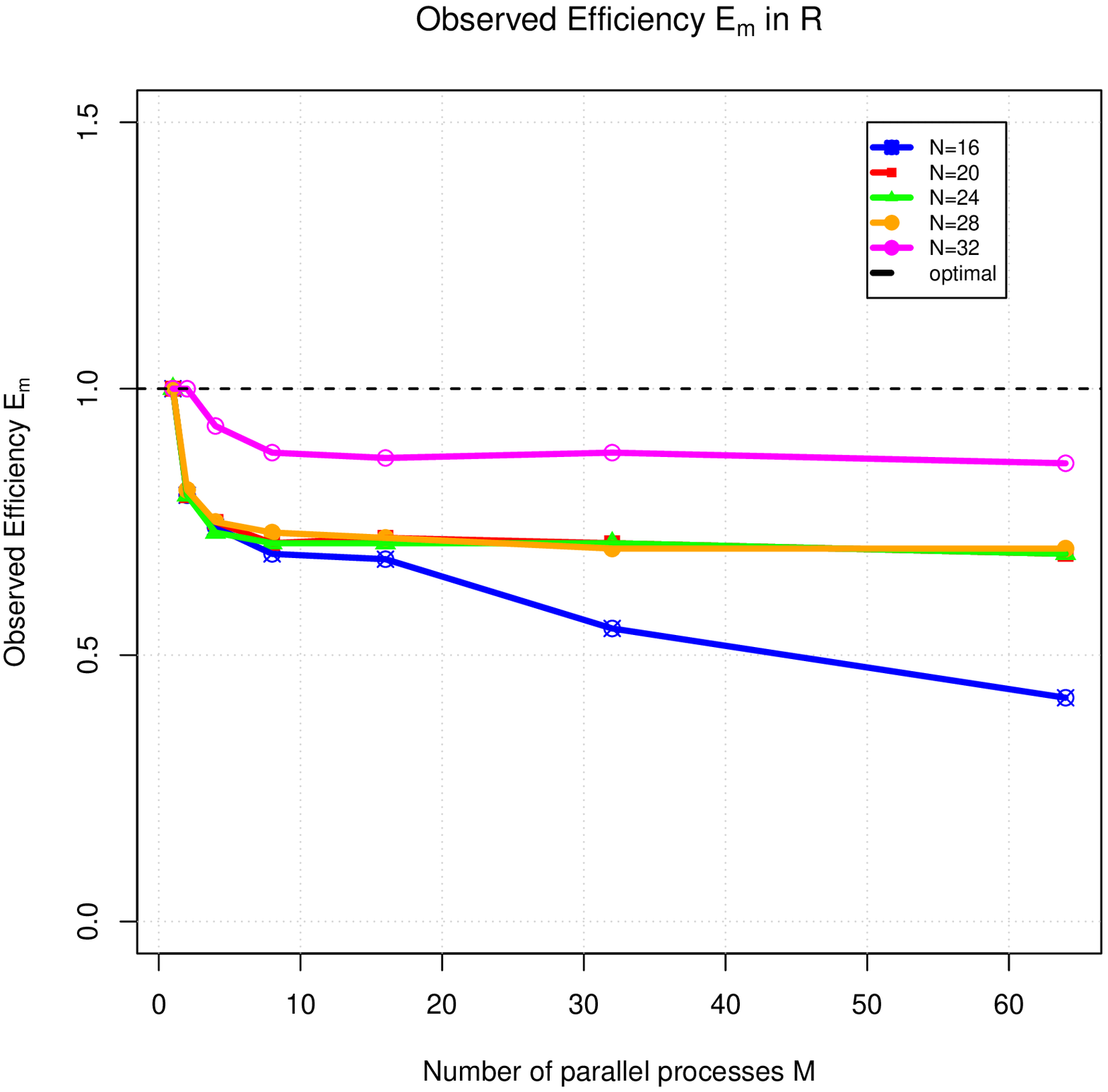} \\
		(a) & (b)
	\end{tabular}
	\caption{(a)~Speedup and (b)~Efficiency in \proglang{R}.}
	\label{fig:speedupR}
\end{figure}

\begin{table}
	\centering
	\caption{Runtime for different number of time steps for \proglang{Julia} implementation.}
	\label{tab:julia_perf}
	\begin{tabular}{rrrrrrrrrrrrrrr}
		\hline
		\multicolumn{8}{l}{(a) Wall clock time in HH:MM:SS} \\
		\multicolumn{1}{r}{\makebox[0.10in][r]{$N$}} &
		\multicolumn{1}{r}{\makebox[0.01in][r]{$M = 1 $}} &
		\multicolumn{1}{r}{\makebox[0.01in][r]{$    2 $}} &
		\multicolumn{1}{r}{\makebox[0.01in][r]{$    4 $}} &
		\multicolumn{1}{r}{\makebox[0.01in][r]{$    8 $}} &
		\multicolumn{1}{r}{\makebox[0.01in][r]{$   16 $}} &
		\multicolumn{1}{r}{\makebox[0.01in][r]{$   32 $}} &
		\multicolumn{1}{r}{\makebox[0.01in][r]{$   64 $}} \\
		\hline
		16&  $<$00:00:01&	$<$00:00:01&	$<$00:00:01&	$<$00:00:01&	$<$00:00:01&	$<$00:00:01&	$<$00:00:01\\
		20&  00:00:07&	00:00:05&	00:00:02&	$<$00:00:01&	$<$00:00:01&	$<$00:00:01&	$<$00:00:01\\
		24&  00:02:23&	00:01:15&	00:00:42&	00:00:22&	00:00:11&	00:00:06&	00:00:03\\
		28&  00:40:08&	00:21:05&	00:11:02&	00:05:56&	00:03:07&	00:01:34&	00:00:51\\
		32&  11:59:54&	06:23:01&	03:24:58&	01:46:06&	00:53:41&	00:27:18&	00:13:38\\
		\hline
		\hline
		\multicolumn{8}{l}{(b) Observed speedup $S_M$} \\
		\multicolumn{1}{r}{\makebox[0.10in][r]{$N$}} &
		\multicolumn{1}{r}{\makebox[0.01in][r]{$M = 1 $}} &
		\multicolumn{1}{r}{\makebox[0.01in][r]{$    2 $}} &
		\multicolumn{1}{r}{\makebox[0.01in][r]{$    4 $}} &
		\multicolumn{1}{r}{\makebox[0.01in][r]{$    8 $}} &
		\multicolumn{1}{r}{\makebox[0.01in][r]{$   16 $}} &
		\multicolumn{1}{r}{\makebox[0.01in][r]{$   32 $}} &
		\multicolumn{1}{r}{\makebox[0.01in][r]{$   64 $}} \\
		\hline
		16&  1.00&	1.57& 2.19&	2.74&	3.36&	4.36&	4.04\\
		20&  1.00&	1.80& 3.16&	6.18&	11.10&	17.65&	25.05\\
		24&  1.00&	1.91& 3.36&	6.37&	13.02&	24.76&	47.79\\
		28&  1.00&	1.90& 3.64&	6.76&	12.90&	25.71&	47.06\\
		32&  1.00&	1.88& 3.52&	6.76&	13.44&	26.37&	52.58\\
		\hline
		\hline
		\multicolumn{8}{l}{(c) Observed efficiency $E_M$} \\
		\multicolumn{1}{r}{\makebox[0.10in][r]{$N$}} &
		\multicolumn{1}{r}{\makebox[0.01in][r]{$M = 1 $}} &
		\multicolumn{1}{r}{\makebox[0.01in][r]{$    2 $}} &
		\multicolumn{1}{r}{\makebox[0.01in][r]{$    4 $}} &
		\multicolumn{1}{r}{\makebox[0.01in][r]{$    8 $}} &
		\multicolumn{1}{r}{\makebox[0.01in][r]{$   16 $}} &
		\multicolumn{1}{r}{\makebox[0.01in][r]{$   32 $}} &
		\multicolumn{1}{r}{\makebox[0.01in][r]{$   64 $}} \\
		\hline
		16&  1.00&	0.78& 0.55&	0.34&	0.21&	0.13&	0.06\\
		20&  1.00&	0.90& 0.79&	0.77&	0.69&	0.55&	0.39\\
		24&  1.00&	0.95& 0.84&	0.79&	0.81&	0.77&	0.75\\
		28&  1.00&	0.95& 0.91&	0.86&	0.81&	0.80&	0.73\\
		32&  1.00&	0.94& 0.88&	0.85&	0.84&	0.82&	0.82\\
		\hline
	\end{tabular}
\end{table}

\begin{figure}
	\centering
	\begin{tabular}{cc}
		\includegraphics[width=0.5\textwidth]{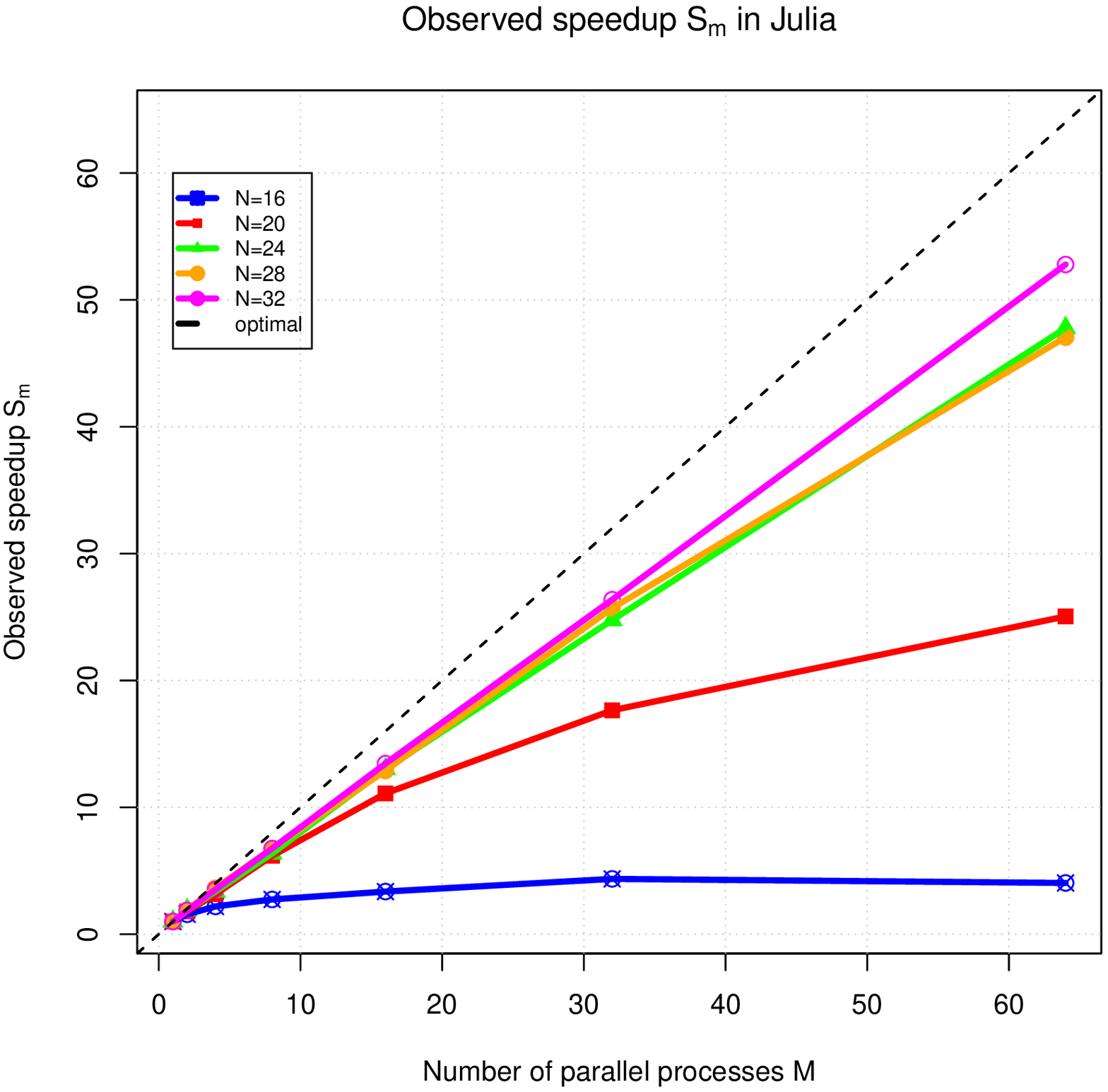} &
		\includegraphics[width=0.5\textwidth]{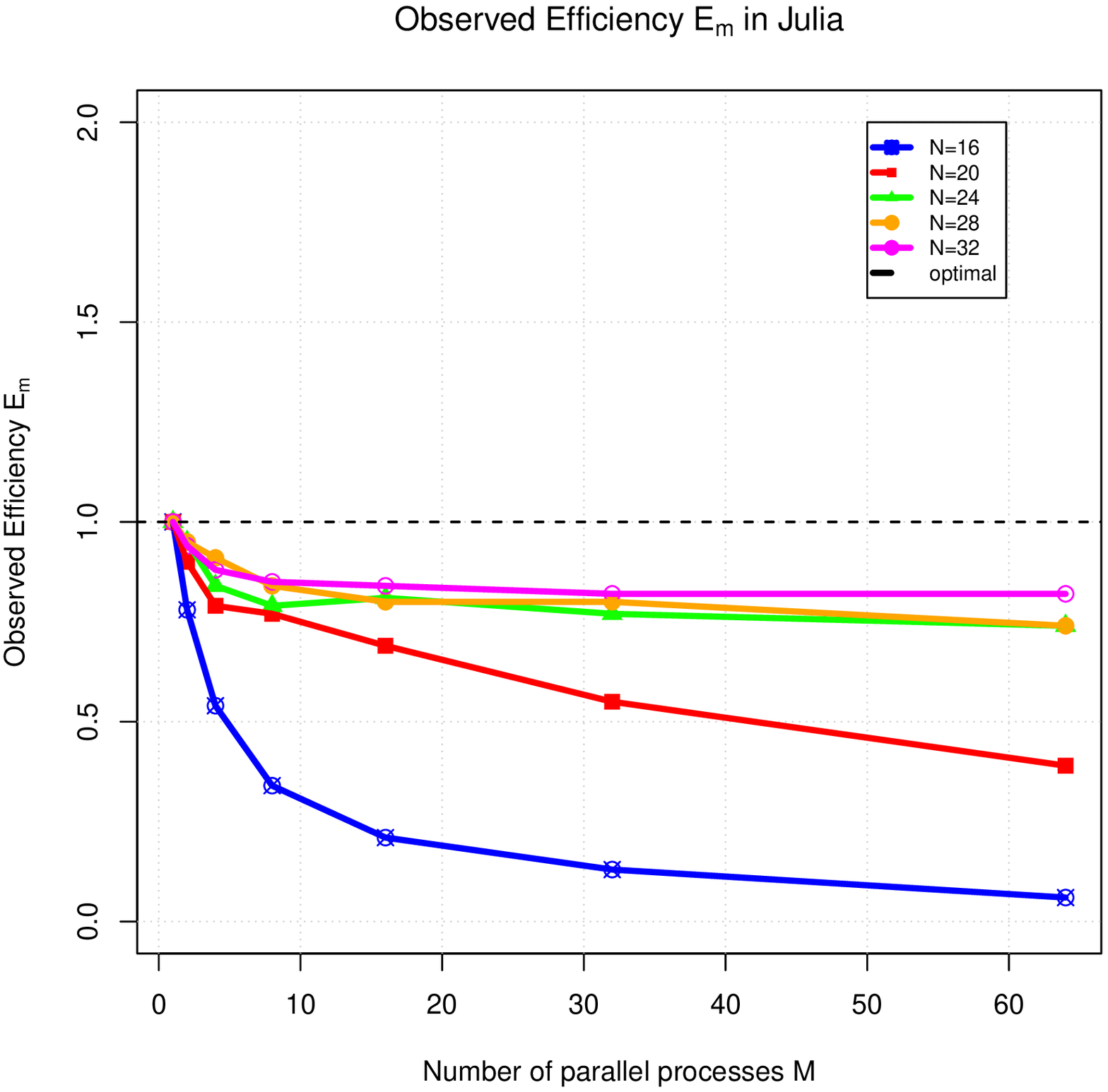} \\
		(a) & (b)
	\end{tabular}
	\caption{(a)~Speedup and (b)~Efficiency in \proglang{Julia}.}
	\label{fig:speedupJulia}
\end{figure}

\begin{table}
	\centering
	\caption{Monte Carlo estimates and corresponding variance estimates for Asian and Fixed Lookback put options with $N=32$. Exact value of the Asian option is $82.115$ and the Lookback option is $93.196$.}
	\label{tab:mc_conv}
	\begin{tabular}{lc|rrrrrrrrrrrrr}
		\hline
		\multicolumn{1}{l}{Option} &
		\multicolumn{1}{l|}{Estimate} &
		\multicolumn{1}{r}{$R = 2^9  $} &
		\multicolumn{1}{r}{$    2^{10}$} &
		\multicolumn{1}{r}{$    2^{11}$} &
		\multicolumn{1}{r}{$    2^{12}$} &
		\multicolumn{1}{r}{$    2^{13}$} &
		\multicolumn{1}{r}{$    2^{14}$} &
		\multicolumn{1}{r}{$    2^{15}$} &
		\multicolumn{1}{r}{$    2^{16}$} \\
		\hline
		Asian & $\hat{V}$ & 82.857&  82.514&	83.425&	83.181&	82.821&	82.615&	82.566& 82.524\\
		& $\widehat{\Var}(\hat{V})$ & 0.735 & 0.362& 0.179& 0.095& 0.050& 0.022& 0.011& 0.006\\
		Lookback & $\hat{V}$ & 93.156&  93.237&	93.312&	93.324&	93.262&	93.236&	93.234&	93.222\\
		& $\widehat{\Var}(\hat{V})$ & 0.022 & 0.009& 0.005& 0.002& 0.001& \textless 0.001& \textless 0.001& \textless 0.001\\
		\hline
	\end{tabular}
\end{table}

\begin{table}
	\centering
	\caption{Partitioned Monte Carlo estimates and corresponding variance estimates for Asian and Fixed Lookback put options with $N=32$ and $R=2^{10}$.}
	\label{tab:var_reduction_mc}
	\begin{tabular}{lc|rrrrrrrrrrrrr}
		\hline
		\multicolumn{1}{l}{Option} &
		\multicolumn{1}{c|}{Estimate} &
		\multicolumn{1}{l}{$ M  =1 $} &
		\multicolumn{1}{r}{$     2 $} &
		\multicolumn{1}{r}{$     4 $} &
		\multicolumn{1}{r}{$     8 $} &
		\multicolumn{1}{r}{$     16$} &
		\multicolumn{1}{r}{$     32$} &
		\multicolumn{1}{r}{$     64$} \\
		&&
		\multicolumn{1}{l}{$R_m = 2^{10}$} &
		\multicolumn{1}{r}{$      2^{9}$} &
		\multicolumn{1}{r}{$      2^{8}$} &
		\multicolumn{1}{r}{$      2^{7}$} &
		\multicolumn{1}{r}{$      2^{6}$} &
		\multicolumn{1}{r}{$      2^{5}$} &
		\multicolumn{1}{r}{$      2^{4}$} \\
		\hline
		Asian & $\hat{V}_s$& 82.077&  82.217&  82.101& 82.232&	81.936&	82.296&	82.165\\
		& $   \widehat{\Var}(\hat{V}_s)$& 0.367& 0.332& 0.315& 0.272& 0.263& 0.212& 0.194\\
		Lookback & $\hat{V}_s$& 93.196& 93.201&	93.171&	93.197&	93.187&	93.216 & 93.205\\
		& $   \widehat{\Var}(\hat{V}_s)$& 0.010& 0.009& 0.008& 0.008& 0.008& 0.007& 0.007\\
		\hline
	\end{tabular}
\end{table}

\begin{table}
	\centering
	\caption{Comparison of the variance estimates from the Partitioned and Shared Sample Monte Carlo methods with $N=32$ and $R=1024$ for Asian put option.}
	\label{tab:smc_hmc}
	\begin{tabular}{lc|rrrrrrrrrrrrr}
		\hline
		\multicolumn{1}{l}{Method} &
		\multicolumn{1}{c|}{Estimate} &
		\multicolumn{1}{r}{$   M = 1$} &
		\multicolumn{1}{r}{$       2$} &
		\multicolumn{1}{r}{$       4$} &
		\multicolumn{1}{r}{$       8$} &
		\multicolumn{1}{r}{$       16$} &
		\multicolumn{1}{r}{$       32$} &
		\multicolumn{1}{r}{$       64$} \\
		\hline
		Partitioned MC & $\hat{V}_s$ & 82.201&  82.160&	82.101& 82.127& 82.109&	82.120&	82.108\\
		& $\widehat{\Var}(\hat{V}_s)$ & 0.367 & 0.220& 0.062& 0.027& 0.005& 0.005& 0.002\\
		Shared Sample MC & $\tilde{V}$ & 82.113&  82.112&	82.120&	82.093&	82.102&	82.124&	82.108\\
		& $\widehat{\Var}(\tilde{V})$ & 0.373 & 0.305& 0.267& 0.203& 0.170& 0.142& 0.123\\
		\hline
	\end{tabular}
\end{table}

\section{Concluding Remarks}
\label{sec:conclusion}
We have presented a method to transform the computation of the expected value in a recombinant binomial tree into an embarrassingly parallel problem by mapping the Bernoulli paths in the tree to the processes on a multiprocessor computer. We also discussed a parallel Monte Carlo estimation method which takes advantage of this partitioning. The methods were implemented both in \proglang{R} and \proglang{Julia}, and were applied to value path-dependent options. Numerical results verify the convergence of the parallel Monte Carlo method and variance reduction with respect to basic Monte Carlo estimation. Performance results indicate that the \proglang{Julia} implementation was significantly faster and more efficient than the \proglang{R} implementation, likely because of the superior handling of loops and the compilation to machine-level code.

\section*{Acknowledgments}
The first author acknowledges financial support from the UMBC High Performance Computing Facility (HPCF) at the University of Maryland, Baltimore County (UMBC). The hardware used in the computational studies is part of HPCF. The facility is supported by the U.S. National Science Foundation through the MRI program (grant no.~CNS--0821258 and CNS--1228778) and the SCREMS program (grant no.~DMS--0821311), with additional substantial support from the University of Maryland, Baltimore County (UMBC). See \url{hpcf.umbc.edu} for more information on HPCF and the projects using its resources.

\bibliographystyle{tfnlm}
\bibliography{binom-references}

\end{document}